\RequirePackage{fix-cm}
\documentclass[smallextended]{svjour3}       % onecolumn (second format)
\smartqed  % flush right qed marks, e.g. at end of proof
\usepackage{graphicx}
\usepackage{makecell}
\usepackage{algorithm} %format of the algorithm 
\usepackage{algorithmic} %format of the algorithm 
\usepackage{multirow} %multirow for format of table 
\usepackage{amsmath} 
\usepackage{xcolor}
\usepackage{appendix}
\usepackage{threeparttable}

\begin{document}

\title{An efficient hybrid hash based privacy amplification algorithm for quantum key distribution}

\titlerunning{An efficient hybrid hash based PA algorithm for QKD}        % if too long for running head

\author{Yan Bingze \and Li Qiong \and Mao Haokun \and Chen Nan}


\institute{Yan Bingze \at
           Harbin Institute of Technology, Harbin 150000, China  \\
                     %  \\
		   \and
		   Li Qiong \at
		   Harbin Institute of Technology, Harbin 150000, China \\
		   \email{qiongli@hit.edu.cn}		  \\
		   \and
		   Mao Haokun \at
		   Harbin Institute of Technology, Harbin 150000, China  \\ 
		   \and
		   Chen Nan \at
		   Harbin Institute of Technology, Harbin 150000, China  \\ 
}

\date{Received: date / Accepted: date}

\maketitle

\begin{abstract}

Privacy amplification (PA) is an essential part in a quantum key distribution (QKD) system, distilling a highly secure key from a partially secure string by public negotiation between two parties. The optimization objectives of privacy amplification for QKD are large block size, high throughput and low cost. For the global optimization of these objectives, a novel privacy amplification algorithm is proposed in this paper by combining multilinear-modular-hashing and modular arithmetic hashing. This paper proves the security of this hybrid hashing PA algorithm within the framework of both information theory and composition security theory. A scheme based on this algorithm is implemented and evaluated on a CPU platform. The results on a typical CV-QKD system indicate that the throughput of this scheme ($261Mbps@2.6\times10^8$ input block size) is twice higher than the best existing scheme ($140Mbps@1\times10^8$ input block size). Moreover, this scheme is implemented on a mobile CPU platform instead of a desktop CPU or a server CPU, which means that this algorithm has a better performance with a much lower cost and power consumption.

\end{abstract}

\section{Introduction}
\label{intro}

% 背景部分

PA is the art of distilling a highly secure key from a partially secure string through public discussion between two parties \cite{Bennett1988a}. The current main application domain of PA is quantum key distribution. 

Quantum key distribution (QKD) is a notable technique which exploits the principle of quantum mechanics to perform the information theoretical secure key distribution between two remote parties, named Alice and Bob \cite{BennettCharlesandBrassard1984}. The performance and practicability of a QKD system have improved rapidly in recent years, and the QKD system has been applied in engineering and commercialization~\cite{Yuan2018,Zhang2019b}. The current trends of a QKD system are improving the performance, such as transmission distance, key rate and the cost reduction~\cite{Liao2017,Xia2019}. A QKD system has two major parts: a quantum optical subsystem for the preparation, transmission and measurement of quantum states; a post-processing subsystem to guarantee the correctness and security of the final secure key \cite{Mao2019}. One of the essential part in the post-processing subsystem is privacy amplification which guarantees the security of the final secure key. 

% 主要矛盾 
The development demand of a QKD system leads to a trilemma in PA as indicated in Fig.~\ref{fig:1}. A PA scheme is supposed to have large block size, high throughput and low cost. However, the trilemma in PA means these objectives cannot be achieved perfectly at the same time. For example, a PA scheme with large block size and high throughput always costs high computing resources. The optimization methods for a PA scheme by changing computing platform and implementation method only improve certain indicator. A really effective optimization method solving the trilemma is to design a new PA algorithm.

%已有PA算法
The new PA algorithm designed in this paper is aimed to combine the advantages of two popular PA algorithms, Toeplitz hashing PA~\cite{Li2019,Liu2016a,Tang2019,XiangyuWangYichenZhangSongYu2016a,Yang2017,Zhang2012a} and modular arithmetic hashing PA \cite{Yan2020,Zhang2014}. The strength of Toeplitz hashing PA is that the input key can be split and handled separately. The strength of modular arithmetic hashing PA is that the input key can be compressed from a binary sequence to a $2^B-$nary sequence. More specifically, the optimal computation complexity of Toeplitz hashing PA is ${\mathop{\rm O}\nolimits} (n\log n)$. The calculation amount of Toeplitz hashing PA with splitting is approximately $C * \left\lceil {\frac{n}{r}} \right\rceil  * r\log r$, where $n$ is the input key length, $r$ is the output key length, and $C$ is a constant. The optimal computation complexity of modular arithmetic hashing PA is the same ${\mathop{\rm O}\nolimits} (n\log n)$.  The calculation amount of modular arithmetic hashing PA with compressing is approximately $C * (n/B)\log (n/B)$. Therefore, this paper aims to design a new algorithm to combine and utilize these two advantages. 
%介绍本文算法

% Fig 1
\begin{figure}
	\includegraphics{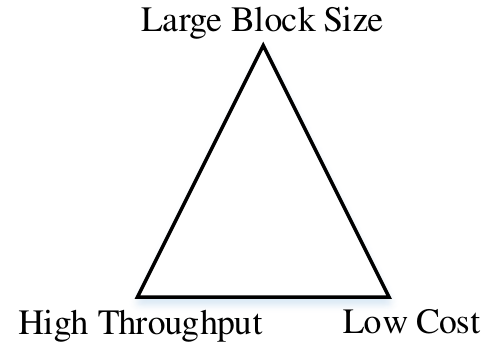}
	\caption{The trilemma in privacy amplification}
	\label{fig:1}       % Give a unique label
\end{figure}

A new PA algorithm based on multilinear-modular-hashing (MMH) and modular arithmetic hashing (MH) is designed to implement the aforementioned optimization. Multilinear-modular-hashing is a well-known universal hashing family for fast message authentication. The most significant feature of this hashing family is that it can not only split and separately handle the input key but also compress the input key from a binary sequence to a $2^B-$ary sequence. The computation complexity of MMH is similarly ${\mathop{\rm O}\nolimits} (n\log n)$. However, the calculation amount of MMH is much lower, approximately $C*\left\lceil {\frac{n}{r}} \right\rceil * (r/B)\log (r/B)$. MMH cannot be used to design PA algorithms directly because the output set of MMH is a fixed finite field instead of a binary sequence of variable length. The output of a PA algorithm is required to be a variable length binary sequence because the secure degree of input partially secure string always changes real-time, especially in a QKD system. To address this issue, modular arithmetic hashing is combined with MMH to design the new PA algorithm, named MMH-MH PA algorithm. The security analysis of MMH-MH PA is accomplished with composition security, and we have confirmed that MMH-MH PA can achieve the same security with Toeplitz hashing PA and modular arithmetic hashing PA. A comparison of advantage and computation of different PA algorithms is indicated in Table~\ref{tab:1}. We assume the calculation amount of MMH-MH PA is $X$, and evaluate the calculation amount of the other PA algorithms to show the advantage of MMH-MH PA.

\begin{table}
\caption{The comparison of computation and advantage of different PA algorithms}
\label{tab:1}       % Give a unique label
\begin{tabular}{cccc}
\hline\noalign{\smallskip}
PA Algorithm  & Advantage & Calculation Amount &  \makecell[c]{Calculation \\Comparison}  \\ 
\noalign{\smallskip}\hline\noalign{\smallskip}
\makecell[c]{Toeplitz \\Hashing PA}  			 & \makecell[c]{spliting the input key \\and handle it separately} 		& $C * \left\lceil {\frac{n}{r}} \right\rceil  * r\log r$ &	$B*X$ \\
\noalign{\smallskip}\hline\noalign{\smallskip}
\makecell[c]{Modular \\Arithmetic \\Hashing PA} 	 & \makecell[c]{transforming the input key \\from binary sequence \\to $2^B-$ary sequence} & $C * (n/B)\log  (n/B)$ 		& $\log  (n/B)/\log(r/B)*X$								 \\
\noalign{\smallskip}\hline\noalign{\smallskip}
MMH-MH PA 						 & combination of both 													& $C*\left\lceil {\frac{n}{r} } \right\rceil * (r/B)\log (r/B)$ &	$X$ \\
\noalign{\smallskip}\hline
\end{tabular}
\end{table}

%算法实验

We design an implementation scheme of MMH-MH PA on CPU platform to evaluate its performance. Because QKD protocols can be divided into discrete variable QKD (DV-QKD) and continuous variable QKD (CV-QKD) according to different working principles~\cite{Yuan2018,Zhang2012a,Zhang2019b}, the key parameters of a typical DV-QKD system~\cite{Yuan2018} and a typical CV-QKD system~\cite{Zhang2019b} are referred for the evaluation. Three conclusions are obtained from the results: 1. the throughput of MMH-MH PA is higher than other PA schemes on CPU. 2. the computing resource cost of MMH-MH PA is lower than existing schemes. 3. the final key rate of a QKD system with MMH-MH PA is better than the common case - $10^8$ input block size PA.

%文章结构
The rest of this paper is organized as follows. The MMH-MH PA algorithm is put forward in Section 2. We analyze the security of the MMH-MH PA algorithm in Section 3. We evaluate the performance of the MMH-MH PA algorithm and analyze the evaluation results in Section 4. Lastly, we draw conclusions about MMH-MH PA algorithm for future work in Section 5.

%第二章
\section{Related Works}
\label{sec:2}
We introduce the definition and advantages of multi-linear modular hashing in this section. First, we discuss the reason why multi-linear modular hashing cannot be directly used for PA. Then we address the problem by introducing modular arithmetic hashing.  
\subsection{Multi-linear Modular Hashing}
The definition of multi-linear modular hashing is indicated as follows:
\paragraph{Definition of Multi-linear Modular Hashing}
Let $p$ be a primer and let $k$ be an integer $k > 0$. Define a family of multi-linear modular hashing functions from $Z^k_p$ to $Z_p$ as follows:
\begin{equation}{\rm{MMH}}: = \left\{ {{{\mathop{\rm g}\nolimits} _a}:Z_p^k \to {Z_p}\left| {a \in Z_p^k} \right.} \right\}\end{equation}
where the function ${{{\mathop{\rm g}\nolimits} _a}}$ is defined for any $a = \left\langle {{a_1}, \cdots ,{a_k}} \right\rangle $, $x = \left\langle {{x_1}, \cdots ,{x_k}} \right\rangle$, ${a_i},{x_i} \in {Z_p}$,
\begin{equation}{{\mathop{\rm g}\nolimits} _a}\left( x \right): = a \cdot x\bmod p = \sum\limits_{i = 1}^k {{a_i}{x_i}\bmod p} \end{equation}

The MMH family is a universal hashing family. Its collision probability $\delta$ is $1/|Z_p|$(See Appendix~\ref{sec:A}), and the proof can be found in \cite{Halevi1997}.

From the structure of the function ${{{\mathop{\rm g}\nolimits} _a}}$ in MMH, three advantages can be found: first, it can split and separately handle the input data like Toeplitz hashing; second, the input key can be compressed from a binary sequence to a $2^B$-nary sequence; third, its main part is large number multiplications, which means it can be accelerated by algorithms such as the Sch\"onhage and Strassen algorithm, GNU multiple precision (GMP) library. However, the main problem with the use of MMH directly for PA is the output set of MMH is a fixed finite field $Z_p$. The expected output form of a PA algorithm is a variable length bit sequence based on the secure degree of the input sequence. To solve this problem, modular arithmetic hashing is introduced to compensate for the drawback of MMH.

\subsection{Modular Arithmetic Hashing}
The definition of multi-linear modular hashing is indicated as follows:
\paragraph{Definition of Modular Arithmetic Hashing}
Let $\alpha$ and $\beta$ be two strictly positive integers, $\alpha > \beta$. Define a family modular arithmetic hashing of functions from $2^\alpha$ to $2^\beta$ as follows:
\begin{equation}{\rm{MH}}: = \left\{ {{h_{b,c}}:{Z_{{2^\alpha }}} \to {Z_{{2^\beta }}}\left| {b,c \in {Z_{{2^\alpha }}}} \right.,\gcd (b,2) = 1} \right\}\end{equation}
where the function $h_{b,c}$ is defined as follows:
\begin{equation}{h_{b,c}}(x): = {{\left( {b \cdot x + c\bmod {2^\alpha }} \right)} \mathord{\left/
 {\vphantom {{\left( {b \cdot x + c\bmod {2^\alpha }} \right)} {{2^{\alpha  - \beta }}}}} \right.
 \kern-\nulldelimiterspace} {{2^{\alpha  - \beta }}}}\end{equation}

Modular Arithmetic Hashing is also a universal hashing, and the output set of modular arithmetic hashing is a variable length bit sequence. Therefore, it can be combined with MMH to design a new PA algorithm.
 
%第三章
\section{MMH-MH PA Algorithm Design and Security Analysis}
The multi-linear modular hashing-modular arithmetic hashing (MMH-MH) PA algorithm is proposed and discussed in this section. Since there is no similar previous PA algorithm, the security analysis of the MMH-MH PA algorithm is very important. The security analysis of the MMH-MH PA proves that it can produce secure key of the same length and provide the same  degree of secure within the framework of both information theory and composition security theory. 

\subsection{MMH-MH PA Algorithm Design}

%算法的主要流程

%MMH-MH PA算法是将MMH和MH算法组合在一起而成的新PA算法。其利用MMH可以对数据分块处理，MH能够输出任意长度比特串的优势。算法可主要分为四步，将输入序列转换到有限域并分块处理，注意这里的决定素域大小的素数建议采用梅森素数，第二步生成随机数选取MMH函数并进行MMH函数的运算，第三步将MMH函数输出的结果输入到MH函数，用随机数构造MH函数

% Fig 2
\begin{figure}
\includegraphics[height=9cm,width=12cm]{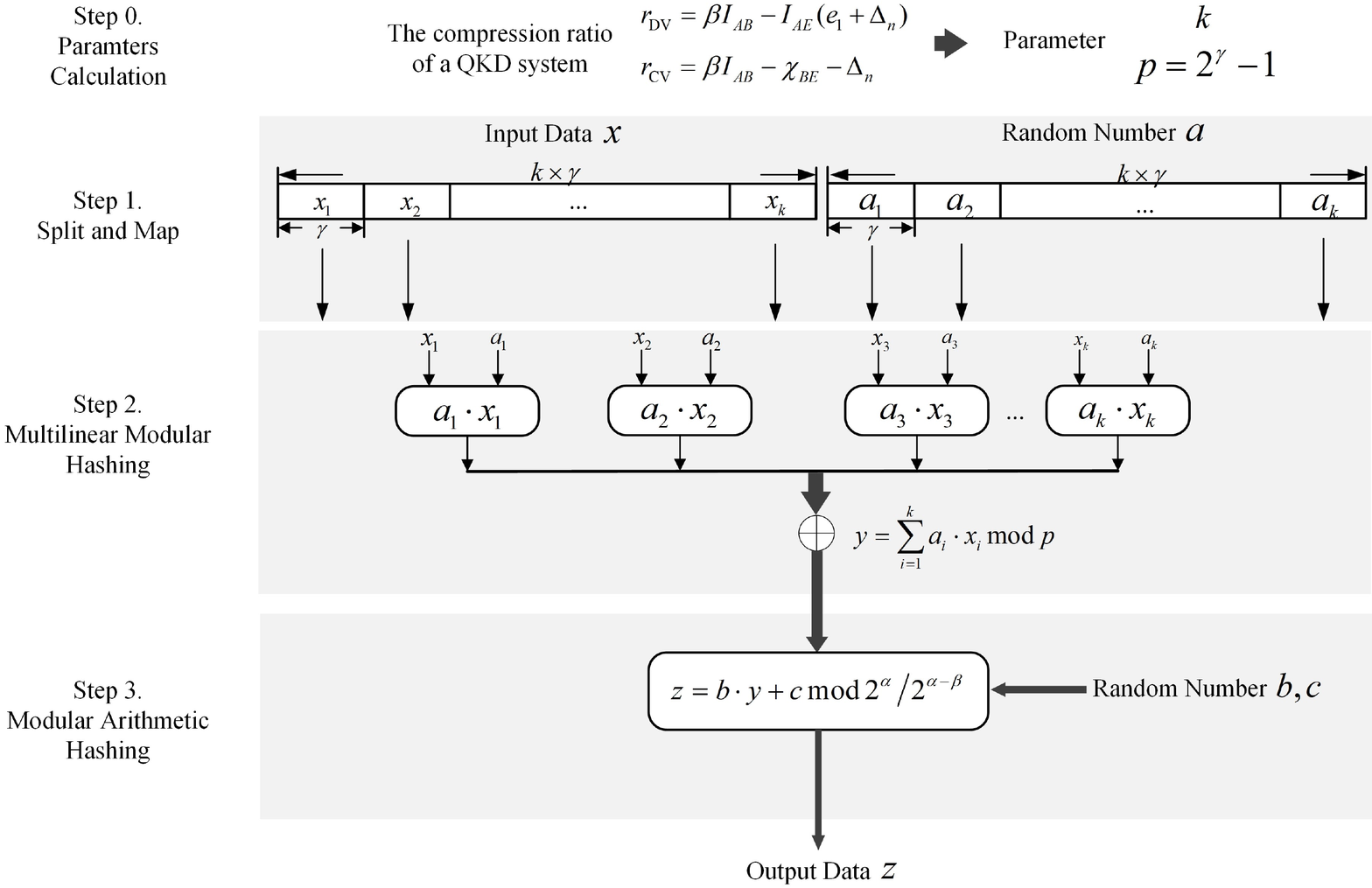}
\caption{The main steps of a MMH-MH PA algorithm}
\label{fig:2}       % Give a unique label
\end{figure} 

The Multilinear Modular Hashing-Modular Arithmetic Hashing (MMH-MH) PA algorithm is a new combined PA algorithm, which splits the input sequence with MMH and outputs a variable length bit sequence with MH.The main steps of the MMH-MH PA algorithm include:
 
0. calculate the parameters of MMH-MH PA algorithm $p$ and $k$;

1. split the input bit sequence and map it to a prime field $Z_p$; 

2. randomly choose function of MMH with input random number $a$ and run the function on the output of step 1; 

3. randomly choose function of MH with input random numbers $b$ and $c$ and run the function on the output of MMH, then output a specified length bit sequence. 

It should be noted that the prime $p$ of the prime field $Z_p$ in step 1 is recommended to be a Mersenne prime. The form of a Mersenne prime is $M_\gamma = 2^\gamma - 1$. All known Mersenne prime is listed in \cite{Merssene}.
Mersenne prime provides two benefits: first, the process of converting input data $x$ from a binary field to a prime field is very simple. The special case is $x_i=2^\gamma - 1$, the data $x_i=2^\gamma - 1$ should be cast away and reloaded. Second, $x\bmod {{\mathop{ M}\nolimits} _\gamma } = \left\lfloor {{x \mathord{\left/
			{\vphantom {x {{2^\gamma }}}} \right.
			\kern-\nulldelimiterspace} {{2^\gamma }}}} \right\rfloor  + x\bmod {2^\gamma }$, so modulus operation can be simplified. 
The main step of the MMH-MH PA algorithm can be further understood in Fig. \ref{fig:2}.
Algorithm~\ref{alg1} described the complete procedure of the MMH-MH PA algorithm.

%算法的规范化表示

\begin{algorithm} %算法开始 

\caption{MMH-MH PA algorithm} %算法的题目 
\label{alg1} %算法的标签 
\begin{algorithmic}[1] %此处的[1]控制一下算法中的每句前面都有标号 
\REQUIRE Input Data: $x \in Z_{2^{k\times\gamma}}$. Random numbers:$a\in Z^k_p$, $b,c \in Z_{2^\gamma}$, $\gcd (b,2) = 1$.  \\\quad\quad//$p=M_\gamma=2^\gamma - 1$ 
\ENSURE $z \in Z_{2^\beta}$ //$\gamma > \beta$%输出结果(此处的ENSURE默认关键字为Ensure在上面已自定义为Output) 

\STATE $x = \left\langle {{x_1}, \cdots ,{x_k}} \right\rangle$ //split data $x$
\STATE $a = \left\langle {{a_1}, \cdots ,{a_k}} \right\rangle$ //split data $a$
\IF{$x_i = 2^\gamma - 1$($i=1,...,k$)} 
\STATE break;   //Reload data $x_i$
\ELSE
\FOR{$i=0$ to $k$}
\STATE $y_i=a_i\times x_i$  
\ENDFOR 
\STATE $y = \sum\limits_{i = 1}^k {{y_i}\bmod p} $  \quad\quad\quad\quad\quad/*MMH function: $y = {{\mathop{\rm g}\nolimits} _a}(x)$*/
\STATE $z = {{\left( {b \cdot y + c\bmod {2^\alpha }} \right)} \mathord{\left/
 {\vphantom {{\left( {b \cdot x + c\bmod {2^\alpha }} \right)} {{2^{\alpha  - \beta }}}}} \right.
 \kern-\nulldelimiterspace} {{2^{\alpha  - \beta }}}}$ /*MH function: $z = {{\mathop{\rm h}\nolimits} _{b,c}}(y)$*/
\ENDIF
\end{algorithmic} 

\end{algorithm}
%使用算法的注意事项

%算法需要进行安全性分析
A significant difference between the MMH-MH PA algorithm and existing PA algorithms is that the MMH-MH PA algorithm utilizes different universal hashing twice as opposed to the previous one-time method. Therefore, it is unclear if security risk exists in the MMH-MH PA algorithm. To verify its security, we analyze the security of the MMH-MH PA algorithm within the framework of information theory and composition security.

%第三章
\subsection{Security Analysis of the MMH-MH PA algorithm}
%为了分析PA算法的安全性，需要首先对窃听者的窃听行为做出假设。本文在经典信息论下对窃听者的窃听行为进行了假设。随后本文证明了MMH-MH PA算法能够输出与已有算法相同长度的安全密钥，并且具有同等的安全性，无论从信息论的角度和可组合安全的角度。
To analyze the security of the new PA algorithm, eavesdropping behavior of the eavesdropper, named Eve, is assumed first under classical information. Then we prove that the MMH-MH PA algorithm can produce secure key of the same length and provide the same degree of security within the framework of information theory and composition security theory.

%窃听者的窃听假设
\subsubsection{Eavesdropping Assumption of the Eavesdropper}
Suppose that $x\in X$ is chosen uniformly at random by two parties, Alice and Bob. The eavesdropper, Eve, is given the value of $w=e(x)$, where $e:X\to W$ is an eavesdropping function. For each $w \in W$, define $c_w=|e^{-1}(w)|$. The probability distribution $p_w$ on $X$, given the value $w$, is the following:
\begin{equation}{p_w}(x) = \left\{ {\begin{array}{*{20}{c}}
{\frac{1}{{{c_w}}}}&{{\rm{if }}\ x \in {e^{ - 1}}(w)}\\
0&{{\rm{if }}\ x \notin {e^{ - 1}}(w)}
\end{array}} \right. \end{equation}
%这里好像需要对输入X和窃听序列的长度做一个假设
	
%MMH-MH PA 的信息论安全性分析
\subsubsection{Information Theory Security Analysis}
\label{sec:3.1}
%基于信息论的安全性分析，是证明PA安全性的最常用方法。它的主要思路是分析在窃听者已知的信息，窃听信息和哈希函数与PA输出之间的互信息。

The use of information theory is the most common method to analyze the security of PA.The secure evaluation standard based on information theory is the mutual information between output of PA $z$ and Eve's information $w$, i.e. ${\mathop{\rm I}\nolimits} (z:w)$. More precisely, it is ${\rm{I}}({q_z}:{e_w},{u_g},{u_h})$, where $q_z=q(z|g,h)$ means the probability distribution $q_z$ induced on $Z$, and the whole expression means the average mutual information over all possible values $w$, MMH functions $g$ and MH functions $h$. It holds that \[{\rm{I}}({q_z}:{e_w},{u_g},{u_h}) = H(q_z) - H(q_z|{e_w},{u_g},{u_h}) \le H(q_z) - {H_{{\rm{Ren}}}}(q_z|{e_w},{u_g},{u_h}),\] where $H$ means the Shannon entropy, and $H_{{\rm{Ren}}}$ means the Renyi entropy (See Appendix \ref{sec:B}).  $H(q_z)$ holds $H(q_z)\le \log_2|Z|$. ${H_{{\rm{Ren}}}}(q|{e_w},{u_g},{u_h})$ can be deduced by two useful lemmas, which have been proved in \cite{Stinson2002}:

\begin{lemma}
\label{lemma:1}
${H_{{\rm{Ren}}}}(p|{u_f}) \ge  - {\log _2}(\delta  + {\Delta _p})$, where $\delta$ means the collision probability of the hashing family (See Appendix \ref{sec:A}).
\end{lemma}
%Then the following relation holds:
%\[{\rm{I}}({q_z}:{e_w},{u_g},{u_h}) = H(q_z) - H(q_z|{e_w},{u_g},{u_h}) \le H(q_z) - {H_{{\rm{Ren}}}}%(q_z|{e_w},{u_g},{u_h})\]
%,~where $H(q_z) - {H_{{\rm{Ren}}}}(q_z|{e_w},{u_g},{u_h})$ can be deduced by the following theorem:

\begin{lemma}
\label{lemma:2}
Suppose $F_1$ is a ${\delta _1} - {\mathop{\rm U}\nolimits} ({D_1};N,{M_1})$ hash family of functions from $X$ to $Y_1$, and $F_2$ is a ${\delta _2} - {\mathop{\rm U}\nolimits} ({D_2};{M_1},{M_2})$ hash family of functions from $Y_1$ and $Y_2$. For any $f_1\in F_1$,$f_2\in F_2$, define $f_1 \circ f2:X\to Y_2$ by the rule  $f_1 \circ f_2(x)=f_2(f_1(x))$. Then \[{f_1 \circ f_2:f_1\in F_1,f_2\in F_2)}\] is a $(\delta_1+\delta_2)-{\mathop{\rm U}\nolimits}(D_1,D_2;N,M_2)$  hash family.
\end{lemma}

\begin{theorem}
\label{theorem:1}
Let $x \in X$ is chosen randomly and $w \in W$ is the output of $w=e(x)$. $g$ is a function chosen by uniform distribution $u_g$ from MMH family $G$ and $h$ is a function chosen by uniform distribution $u_h$ from MH family $H$. The collision probability of $G$ and $H$ is $\delta_g$ and $\delta_h$. Perform $y = g(x)$ and $z = h(y)$, then the following relation holds \[{H_{{\rm{Ren}}}}({q_z}|{u_g},{u_h},{e_w}) \ge  - {\log _2}(\frac{|Z|}{|X|} + {\delta _g} + {\delta _h}),\] where $q_z$ is the probability distribution of $z \in Z$.
\end{theorem}
\begin{proof}
	For all $w \in W$, $p_{x|w}$ is an uniform distribution, and $\Delta_{p_{x|w}}=|W|/|X|$. When lemma \ref{lemma:1} is applied, it holds that \[{H_{{\rm{Ren}}}}({q_z}|{u_f},{e_w}) \ge  - {\log _2}(\frac{{|W|}}{{|X|}} + {\delta _f}).\] And when lemma \ref{lemma:2} is applied, the collision probability of the composition construction with MMH and MH is $\delta_g + \delta_h$. Then it is clear that ${H_{{\rm{Ren}}}}({q_z}|{u_g},{u_h},{e_w}) \ge  - {\log _2}(\frac{|W|}{|X|} + {\delta _g} + {\delta _h})$.
\end{proof}

According to theorem \ref{theorem:1}, the mutual information holds that \[{\rm{I}}({q_z}:{e_w},{u_g},{u_h}) \le \left| Z \right| + {\log _2}(\frac{{|W|}}{{|X|}} + {\delta _g} + {\delta _h}).\]

Suppose that $x$,$w$,$z$ is all binary form, i.e. $|X|=2^n$, $|Y|=2^m$, $|Z|=2^t$, And $\delta_g={1\mathord{\left/{\vphantom {1 {{2^\gamma }}}} \right. \kern-\nulldelimiterspace} {{2^\gamma }}}$, ${\delta _h} = {1 \mathord{\left/ {\vphantom {1 {{2^m}}}} \right.\kern-\nulldelimiterspace} {{2^m}}}$. Then the mutual information holds that
\[{\rm{I}}({q_z}:{e_w},{u_g},{u_h}) \le {{\left( {{2^{t - n + m}} + {2^{m - p}}} \right)} \mathord{\left/
 {\vphantom {{\left( {{2^{t - n + m}} + {2^{m - \gamma}}} \right)} {\ln 2}}} \right.
 \kern-\nulldelimiterspace} {\ln 2}}.\]
When $m=n-t-s$, $s$ is called security coefficient in information theory secure analysis in PA. And let $m << \gamma$ in MMH-MH PA design, the following relation holds,
\[{\rm{I}}({q_z}:{e_w},{u_g},{u_h}) \le {{{2^{ - s}}} \mathord{\left/
 {\vphantom {{{2^{ - s}}} {\ln 2}}} \right.
 \kern-\nulldelimiterspace} {\ln 2}}.\]
This is a well-established conclusion in the previous PA algorithm security analysis. It shows that MMH-MH PA has the equivalent security with the common PA algorithm (e.g. Toeplize-based PA and modular arithmetic hashing PA) under information theory.

%MMH-MH PA 的可组合安全性分析
\subsubsection{Composition Security Analysis}
Composition security is a popular security analysis tool in QKD, so we use it to analyze the security level of our PA algorithm. The secure evaluation standard based on security analysis is the Statistical distance (Definition \ref{definition:1}) between the conditional probability distribution of PA output given the eavesdropping information $w$ and the uniform distribution of perfect key, as defined in Definition \ref{definition:2}.
\begin{definition}
\label{definition:1}
Let $p$ and $q$ be two probability distributions on the set $X$. The statistical distance between $p$ and $q$, denoted $\mathop{\rm d}(p,q)$, is defined as follows:
\[{\rm{d}}(p,q){\rm{ = }}\frac{1}{2}\sum\limits_{x \in X} {\left| {p(x) - q(x)} \right|}. \]
\end{definition}

\begin{definition}
\label{definition:2}
Let $p_x$ be a probability distribution of random variable $x$ on the set $X$, $p_w$ be a probability distribution of random variable $w$ on the set $W$ and $u$ be uniform distribution, and let $\epsilon > 0$. $x$ is said to be $\epsilon$-secure with respect to $p_w$ if $d(p_{x|w},u)\le \epsilon$.
\end{definition}

With the proof in Section~\ref{sec:3.1}, the upper bound of collision probablity ${\Delta _{{q_z}|{u_g},{u_h},{e_w}}}$ can be obtained easily, \[{\Delta _{{q_z}|{u_g},{u_h},{e_w}}} \le {2^{ - m}} + {2^{t - n}}{\rm{ + }}{2^{{\rm{ - }}\gamma }}.\]

Then we give the relationship between the collision probability and the statistical distance with uniform distribution of a probability distribution with the following lemma, which is proved in \cite{Stinson2002}:

\begin{lemma}
\label{lemma:3}
Let $(Y,p)$ be a probability space. Then \[{\mathop{\rm d}\nolimits} (p,u) \le \sqrt {{\Delta _p}|Y| - 1} /2.\]
\end{lemma}

Then the following relation holds
\[{\mathop{\rm d}\nolimits} ({q_z}|{u_g},{u_h},{e_w},u) \le \sqrt {{2^{m + t - n}} + {2^{m - \gamma}}} /2.\]
Similarly, when $m=n-t-s$, and let $m \ll \gamma$ when MMH-MH PA is designed, the following relation holds
\[\varepsilon  \le {2^{\frac{{ - s}}{2} - 1}}.\]
This means the output key of MMH-MH PA can be ${2^{\frac{{ - s}}{2} - 1}}$-secure under composition security, when output lentgh is $m=n-t-s$.

%第四章
\section{Implementation and Simulation}
\label{sec:4}

To evaluate the actual performance of the MMH-MH PA algorithm, we used the parameters of a typical DV-QKD system \cite{Yuan2018} and a typical CV-QKD system \cite{Zhang2019b} to design and evaluate the MMH-MH PA scheme, and the parameters are listed in Table~\ref{tab:2}.

\begin{table}
	\caption{Parameters of a typical DV-QKD system and a typical CV-QKD system}
	\label{tab:2}       % Give a unique label
	\begin{threeparttable}          %这行要添加
	\begin{tabular}{c|c|c|c}
		\hline
		\hline
		\multicolumn{2}{c|}{A Typical DV-QKD System\cite{Yuan2018}}  & \multicolumn{2}{|c}{A Typical CV-QKD System\cite{Zhang2019b}} \\
		\hline
		\hline
		Parameters & Values & Parameters & Values \\
		\hline
		system clock frequency $f$ & $1\ GHz$ & transmission distance $l$ & $50\ km$ \\
		transmission distance $l$ & $10\ km$ & excess channel noise $V_{\epsilon}$ & $0.005$ \\
		 $u,v,w$\tnote{1} & $0.4,\ 0.1,\ 0.0007$ & electric noise $V_{el}$ & $0.041$ \\
		 $p_{u},p_{v},p_{w}$\tnote{2} & $96.9\%,\ 1.6\%,\ 1.4\%$ & $\eta$\tnote{5} & $0.606$ \\
		 $\eta_{\delta}$\tnote{3} & $22.5\%$  & reconciliation efficiency $\beta$ & $0.95$ \\
		SPD dark count rate $Y_{0}$ & $4.5^{-6}$ & security parameter of PA $\epsilon$ & $10^{-10}$ \\
		SPD system error $e_{\delta}$ & $3\%$ &  &    \\
		 $p_{Z}$, $p_{X}$\tnote{4} & $96.7\%$ $3.4\%$ &  &  \\
		security parameter of PA $\epsilon$ & $10^{-10}$ & & \\
		\hline
		\hline
	\end{tabular}
	\begin{tablenotes}    %这行要添加， 从这开始
		\footnotesize               %这行要添加
		\item[1] $u,v,w$: photon fluxes for signal, decoy and vacuum pulses         %这行要添加
		\item[2] $p_{u},p_{v},p_{w}$: respective probabilities for signal, decoy and vacuum pulses        %这行要添加
		\item[3] $\eta_{\delta}$: photon detector (SPD) detection efficiency 
		\item[4] $p_{Z}$, $p_{X}$: Z and X basis probabilities 
		\item[5] $\eta$: homodyne detector efficiency 
	\end{tablenotes}            %这行要添加
	\end{threeparttable}          %这行要添加
\end{table}

%由于QKD系统的PA保留比是MMH-MH PA方案设计的重要参考，我们对这两个系统的PA保留比进行了计算。随后我们进行了MMH-MHPA方案的设计，并从三个方面测试了其性能并与已有方案进行了比较：1、方案在不同码长下的处理速率；2、方案对于系统最终成码率的影响；3、方案对于计算资源的消耗。
To design a MMH-MH PA scheme for a specific QKD system, two key parameters need to be confirmed. One is the unit block size $\gamma$, and the other is the block number $k$. The input block size of PA $n$ is equal to $\gamma \times k$ and expected to be as large as possible. Raising the unit block size $\gamma$ increases the computation cost and decreases the throughput.
The block number $k$ is expected to be as large as possible, but it is restricted by the compression ratio of the QKD system. So we calculated the compression ratio of QKD systems for evaluation to confirm the block number $k$. Then we designed the MMH-MH PA schemes with different unit block size $\gamma$ and evaluated their performance in three aspects: 1. throughput at different input block sizes; 2. the computational resource consumption of this scheme; 3. the final secure key rate of a system with this scheme.

As previously mentioned, the block number $k$ is restricted to the compression ratio of the QKD system. The compression ratio of a DV-QKD system can be calculated by $r = \beta {I_{AB}} - {I_{AE}(e_1+\Delta_n)} $ (See \cite{Yuan2018} for more details) and the compression ratio of a CV-QKD system can be calculated by $r = \beta {I_{AB}} - {\chi _{BE}} - \Delta_n$ (See \cite{Milicevic2017} for more details). Then the block number $k$ of the MMH-MH PA scheme can be confirmed according to the compression ratio, $k \le \frac{1}{r}$. The results are listed in Table~\ref{tab:3}.

\begin{table}
\caption{MMH-MH PA Parameter $k$ for typical QKD systems}
\label{tab:3}      % Give a unique label
\begin{tabular}{ccc}
		\hline
		\hline
		System & Compression Ratio $r$ & Block Number of MMH-MH PA $k$ \\
		\hline
		Typical DV-QKD system \cite{Yuan2018} & 0.2957 & 3 \\
		Typical CV-QKD system \cite{Zhang2019b}& 0.0972 & 10 \\
		\hline
		\hline
	\end{tabular}
\end{table}

First, we evaluated the throughput of MMH-MH PA schemes with different block size $n=k \times \gamma$ on typical QKD systems. We compared these results with existing PA schemes on similar platforms in Fig.~\ref{fig:3}, and key parameters of these schemes are listed in Table~\ref{tab:6}. The experiment results indicate that the existing best PA scheme is implemented on a high performance CPU i9-9900k, and our scheme reaches nearly twice throughput on a common mobile CPU i5-7300HQ. So our scheme is able to reach higher throughput with lower cost. More specifically, the lower cost of our scheme can be reflected on: 1. the low cost CPU platform (fewer cores and lower basic frequency) used for the scheme; 2. the lower memory resource cost and the lower actual memory used. The higher performance and lower cost of our scheme verify the computing advantages of the MMH-MH PA algorithm.

% Fig 3
\begin{figure}
\includegraphics[height=7.5cm,width=12.5cm]{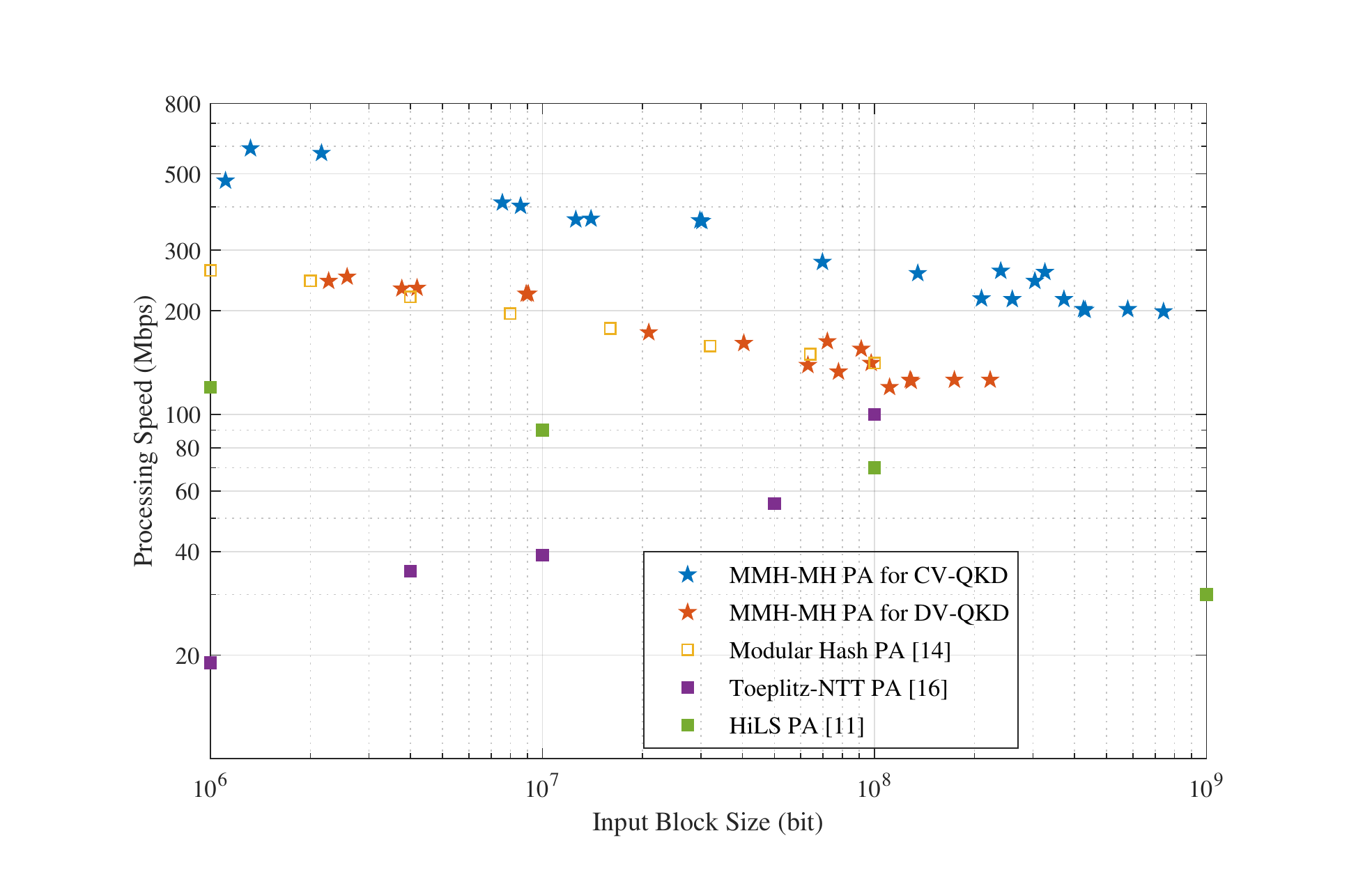}
\caption{The throughput comparison between MMH-MH PA scheme and exiting schemes}
\label{fig:3}       % Give a unique label
\end{figure}

\begin{table}[htbp]
	\renewcommand{\arraystretch}{1.3}
	\newcommand{\tabincell}[2]{\begin{tabular}{@{}#1@{}}#2\end{tabular}}
	\caption{Key parameters of the compared schemes}
	\label{tab:6}
	\centering
	\resizebox{\textwidth}{15mm}{
	\begin{tabular}{cccccccc}
		\hline
		\hline
		PA scheme & Hash & Based Method & Platform & \tabincell{c}{Number \\ of cores} & \tabincell{c}{Basic \\ Frequency} & Memory & \tabincell{c}{Actual \\ Memory \\Used}  \\
		\hline
		MMH-MH PA & \tabincell{c}{MMH \\+MH} & GMP multiply & Intel i5-7300HQ & 4 & 2.20 GHz & 8GB & 120MB\\
		\tabincell{c}{Modular Hash PA}\cite{Yan2020} & \tabincell{c}{Modular \\Hash} & GMP multiply & Intel i9-9900k & 8 & 3.60GHz & 16GB & 220MB\\ 
		\tabincell{c}{Toeplitz NTT PA \cite{Yuan2018}}   & Toeplitz & NTT & \tabincell{c}{Intel Xeon E5-2620v2\\+ Intel Xeon Phi 7120A} & \tabincell{c}{8 \\ +61} &  \tabincell{c}{2.10GHz \\ +1.24GHz} & \tabincell{c}{16GB \\ +128GB} & -- \\
		HiLS PA\cite{Tang2019} & Toeplitz & FFT & Intel E5-2640 v3 $\times$ 2 & 8$\times$2 & 2.10GHz & 128GB& --\\
		\hline 
		\hline
	\end{tabular}}
\end{table}

The following simulation experiment shows the influence of our scheme on the final key rate of a QKD system. When the scheme throughput can satisfy the system demand, the main factor on the final key rate in PA is the input block size $n$. In our scheme, the input block size is $n = k \times \gamma$. It can reach $2 \times 10^8$ in the typical DV-QKD system and $8 \times 10^8$ in the typical CV-QKD system. In most QKD systems, $n = 10^8$ is common. We compare the final key rates of the QKD system with our scheme, the infinite size and the common input size in Fig. \ref{fig:4}. The final key rate of the QKD system with our scheme is better than the usual size. The improvement is more obvious in the CV-QKD system because the finite-size-effect is more serious in CV-QKD. This result can also be reached with the scheme in \cite{Tang2019}, while our scheme saves more computing resource and avoids the possible truncation error caused by floating point FFT of the scheme in \cite{Tang2019}.   

According to the above simulation experiment, the scheme with the MMH-MH PA algorithm shows better performance, lower cost and a positive effect on the final key rate of the entire system. Among these advantages, we regard low cost as the most competitive advantage for the near-term QKD system. It can improve the system performance and practicability in the resource consumption sensitivity scene, such as chip-based QKD systems.

% Fig 4
\begin{figure}
\includegraphics[height=6cm,width=12.5cm]{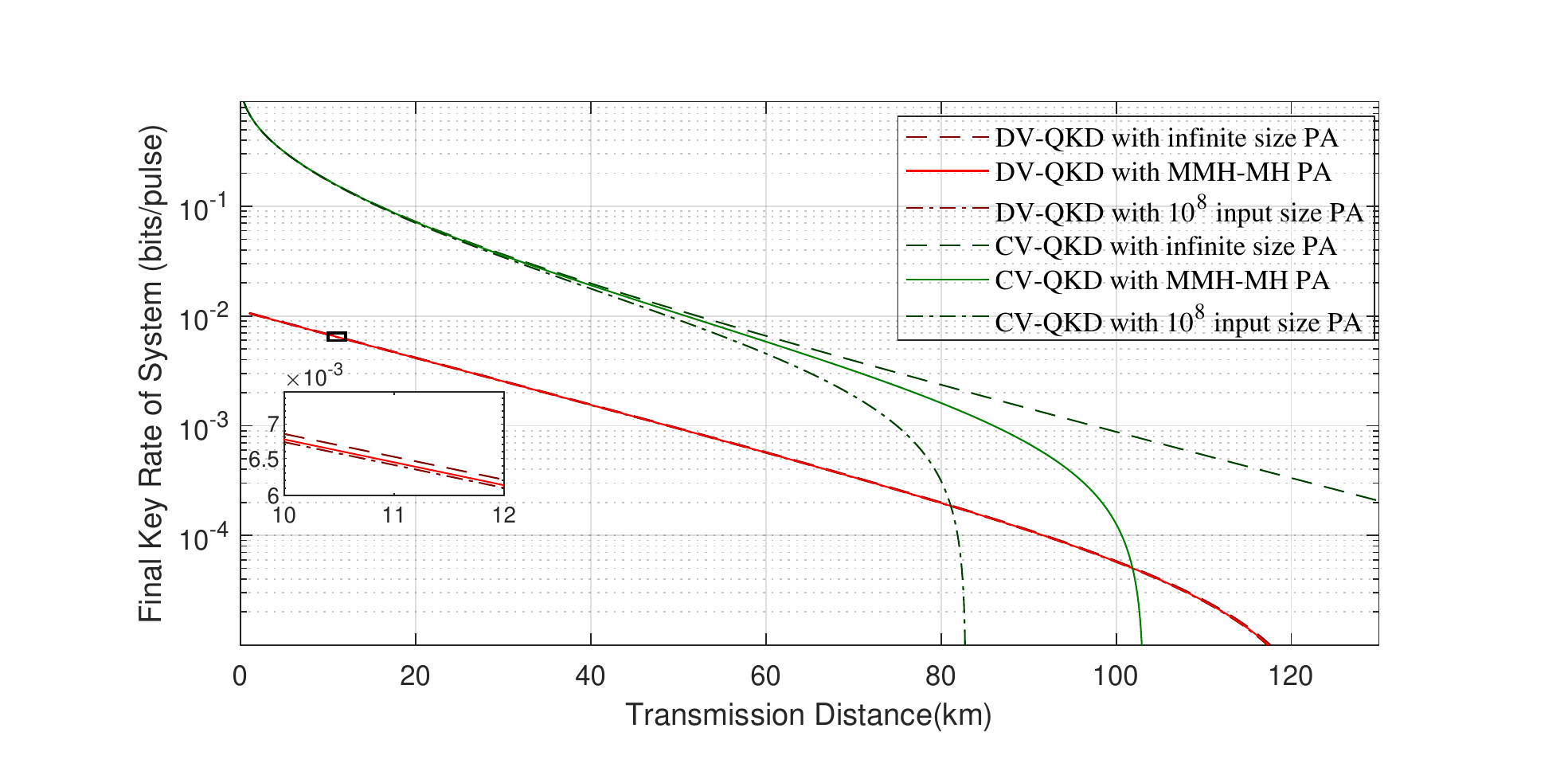}
\caption{the Final Key Rate of QKD Systems with MMH-MH PA}
\label{fig:4}       % Give a unique label
\end{figure}
\section{Conclusion}

In this research, a new PA algorithm, named MMH-MH PA, is proposed to improve the overall performance of PA. This algorithm not only performs better but also consumes fewer computing resources. The experiment results show that the MMH-MH PA algorithm: 1. provides higher throughput at larger block size than existing CPU PA algorithms, especially on a CV-QKD system; 2. the final key rate of a system with the MMH-MH PA algorithm is extremely close to the one with an infinite size PA; 3. the MMH-MH PA algorithm costs fewer computing resources. We regards the third point as the most outstanding advantage of MMH-MH PA for current QKD systems. It provides higher final key rate and lower cost for a QKD system. It is particularly important for the resource sensitive QKD systems, such as a chip-based QKD system. In the future, we will design a FPGA-based PA scheme based on the MMH-MH PA. We believe this scheme will significantly improve the performance of a chip-based QKD system. 

\begin{acknowledgements}
	
This work was supported in part by the National Natural Science Foundation of China under Grant No. 62071151, 61301099. 

\end{acknowledgements}



%\begin{thebibliography}{}

\bibliographystyle{spmpsci}

\bibliography{library}

\appendix
\section{Universal Hashing Family}
\label{sec:A}
A $(D;N,M)$ hashing family is a set $F$ of D functions that $f:X \to Y$ for each $f \in F$, $|X|=N$ and $|Y|=M$.

A $(D;N,M)$ hashing family $F$ is $\delta$-universal hashing means for two distinct elements $x_1,x_2 \in X$,there exists at most $\delta D$ functions $f \in F$ such that $f(x_1)=f(x_2)$. The parameter $\delta$ is the collision probability of the hash family.
%总体叙述I（Z:W）的求法

\section{Renyi Entropy and Collision Probability}
\label{sec:B}
Let $(X,p_x)$ be a probability space. The Renyi entropy of $(X,p_x)$, denoted $H_{\rm{Ren}}(p_x)$, is defined to be \[{H_{{\mathop{\rm Re}\nolimits} n}}({p_x}) =  - {\log _2}{\Delta _{{p_x}}}\]
where $\Delta _{{p_x}}$ denotes the collision probability of the probability distribution $p_x$, is defined by \[{\Delta _{{p_x}}} = \sum\limits_{x \in X} {{{\left( {p(x)} \right)}^2}} .\]
A property of the Renyi entropy is useful in this paper:
\begin{lemma} Let $(X,p_x)$ be a probability space. ${H_{{\mathop{\rm Ren}\nolimits} }}(p_x) \le H(p_x)$.
\end{lemma}

\end{document}